\documentclass[lettersize,journal]{IEEEtran}
\usepackage{amsmath,amsfonts}
\usepackage{amsthm, amssymb}
\usepackage[ruled,linesnumbered]{algorithm2e}
\usepackage{array}
\usepackage[caption=false,font=normalsize,labelfont=sf,textfont=sf]{subfig}
\usepackage{textcomp}
\usepackage{stfloats}
\usepackage{url}
\usepackage{bm}
\usepackage{verbatim}
\usepackage{graphicx}
\usepackage{cite}

\usepackage{amsmath,amssymb,amsfonts}
\usepackage{graphicx}
\usepackage{textcomp}
\usepackage{xcolor}
\usepackage{cite}  % 改善引用格式

\newtheorem{lemma}{Lemma}

\begin{document}

\title{Stable-MoE: Lyapunov-based Token Routing for Distributed Mixture-of-Experts Training over Edge Networks}

\author{Long  Shi,
Bingyan Ou,
Kang Wei,
Weihao Zhu,
Zhe Wang,
and Zhiyong Chen
% Wei Zhang
        % <-this % stops a space
\thanks{Long Shi, Bingyan Ou, and Weihao Zhu are with the School of Electronic and Optical Engineering, Nanjing University of Science and Technology, Nanjing 210094, China (e-mail:$\lbrace$longshi, oubingyan, zhuwh0813$\rbrace$@njust.edu.cn). % <-this % stops a space
Kang Wei is with the School of Cyber Science and Engineering, Southeast University, Nanjing 211189, China (e-mail: kang.wei@seu.edu.cn). 
Zhe Wang is with the School of Computer Science and Engineering, Nanjing University of Science and Technology, Nanjing 210094, China (e-mail: zwang@njust.edu.cn).
Zhiyong Chen is with the Cooperative Medianet Innovation Center and the Department of Electronic Engineering, Shanghai Jiao Tong University, Shanghai 200240, China (e-mail: zhiyongchen@sjtu.edu.cn). 
}}
% \thanks{Wei Zhang is with the School of Electrical Engineering and Telecommunications, The University of New South Wales, Sydney NSW 2052, Australia (email: w.zhang@unsw.edu.au).}



\maketitle

\begin{abstract}
The sparse activation mechanism of mixture of experts (MoE) model empowers edge intelligence with enhanced training efficiency and reduced computational resource consumption. 
However, traditional token routing in distributed MoE training faces significant challenges in resource-constrained edge networks characterized by heterogeneous computing capabilities and stochastic token arrivals, which inevitably suffer from workload backlog, resource inefficiency, and performance degradation. 
To address this issue, we propose a novel Lyapunov-based token routing framework for distributed MoE training over resource-heterogeneous edge networks, termed Stable-MoE. 
Specifically, we formulate a stochastic optimization problem to maximize both system throughput and gating consistency via optimizing the token routing strategy and computational resource allocation, while ensuring long-term stability of both token and energy queues at the edge devices. 
Using the Lyapunov optimization, we transform the intractable long-term optimization problem into tractable per-slot subproblems by enabling online decision-making of token routing and computation frequency utilization without the knowledge of future system states. 
Experimental results on the SVHN and CIFAR-100 datasets demonstrate that Stable-MoE outperforms the baselines with at least $40\%$ and $5\%$ gains in system throughput and test accuracy, respectively. 
\end{abstract}

\begin{IEEEkeywords}
Distributed MoE, token routing, resource heterogeneous, Lyapunov optimization, edge network.
\end{IEEEkeywords}

\section{Introduction}
\IEEEPARstart{M}{ixture} of experts (MoE) models have recently gained significant attention as an effective architecture for large-scale neural networks, especially in distributed training systems where experts are trained or deployed on edge devices~\cite{Zhu2025Trustworthy,xue2025wdmoe,pang2026low, deng2025federated}. %,Zhang2025Communication
The main advantage of MoE lies in its sparse activation mechanism, which selectively activates only a subset of experts for each token input. 
This design significantly reduces computational load and enhances the scalability of the model~\cite{fedus2022switch}.
A central module of MoE is token routing, which determines the assignment of each token input to the most suitable experts. It is well established that an efficient routing strategy is crucial for optimizing the utilization of computational resources in MoE while simultaneously satisfying the demand of high-quality learning performance~\cite{shazeer2017outrageously}. 

In this context, various token routing strategies have been developed in centralized MoE architectures. For example, a pioneering work of MoE in~\cite{shazeer2017outrageously} employs a simple top-$K$ gating mechanism that routes each token to $K$ experts with the highest gating scores based on token feature. Later, Switch Transformer of~\cite{fedus2022switch} simplifies this approach by routing each token to the best expert only, while GLaM of~\cite{du2022glam} uses the top-$2$ routing strategy combined with auxiliary load balancing losses. GShard of~\cite{lepikhin2021gshard} also adopts the top-$2$ routing strategy but enforces expert capacity constraints by discarding tokens that exceed the per-expert token quota. Specifically, this quota refers to the maximum number of tokens that each expert is allowed to process per batch. 
% Moving forward, BASE Layers of~\cite{lewis2021base} reformulates routing as a balanced assignment problem to ensure that tokens are evenly distributed across experts without being discarded. 
Recent works such as Expert Choice Routing of~\cite{zhou2022mixture} address the issues of expert utilization by inverting the traditional routing paradigm: instead of tokens being routed to experts based on their compatibility, experts actively select tokens that align with specialized capabilities. 

% As the deployment of MoE becomes promising in edge networks for distributed training, traditional token routing strategies face new challenges specific to resource-heterogeneous environments~\cite{zhu2025moe,Zhu2025Randomized}. 
As MoE becomes promising in distributed edge intelligence, its implementation faces new challenges in terms of trustworthiness~\cite{zhu2025moe}, security~\cite{cao2022magsign, huang2024manipulating}, and privacy~\cite{wang2023security}. Beyond these, routing decisions of MoE over wireless networks are required to account for diverse resource constraints and random channel conditions. For instance, the work of~\cite{xue2025wdmoe} proposes WDMoE to mitigate delays in wireless networks that considers end-to-end latency into token routing. Beyond the latency-aware routing, recent works such as SMILE of~\cite{he2022smile} explore bi-level routing to adapt to heterogeneous network bandwidths, where tokens are first routed to a computing node and then to a specific GPU within that node. Later, the work of~\cite{song2025mixture} proposes channel-aware gating functions that incorporate wireless communication quality into token routing for edge computing scenarios. 
Despite the aforementioned routing strategies, a critical end-to-end throughput bottleneck manifests in practical deployments: under dynamic token arrival, token backlogs may occur at the edge devices with heterogeneous computational capacities, which in turn results in pronounced workload backlog and performance degradation. To date, how to address this token backlog issue remains an open and non-trivial research problem. 
% However, these edge-specific routing adaptations typically focus on individual aspects of the problem and lack a systematic framework for jointly optimizing token routing under the complex constraints of heterogeneous edge networks with stochastic workloads.

In this paper, we consider a distributed MoE training framework over resource-constrained edge networks under heterogeneous computational capabilities and stochastic token arrivals.
To enable efficient MoE training under such conditions, we propose a novel Lyapunov-based token routing framework for distributed MoE training over edge networks, termed Stable-MoE. 
In particular, we formulate a stochastic optimization problem that jointly maximizes system throughput and gating consistency while ensuring long-term stability of both token and energy queues. 
By leveraging the Lyapunov optimization theory, the long-term stochastic problem is decomposed into tractable per-slot subproblems, which are solved to obtain the optimal token routing matrix and computation frequency allocation at each time slot. 
This design not only guarantees queue stability but also balances workload across heterogeneous edge servers, thereby enhancing both model performance and training efficiency.

\section{System Model And Problem Formulation}
\subsection{System Model}
Fig. \ref{fig_1} illustrates a distributed MoE architecture over an edge network consisting of a router and $J$ heterogeneous edge servers, wherein the gating network is executed in the router and different experts are distributed over different edge servers~\cite{xue2025wdmoe}. 
The distributed MoE training over edge networks proceeds as follows:

(1) Upon receiving a set of tokens, the gating network activates a group of experts using a top-$K$ routing strategy, and routes each token to the selected experts for forward propagation computing. 

(2) Each edge server employs its locally deployed experts to process the routed subset of tokens, generates the forward outputs, and sends the outputs to the router. 

(3) The router aggregates the outputs from all the experts to produce the final forward result, and computes the training loss for the tokens. 

(4) The router broadcasts the training loss to the edge servers, after which the gating network and experts independently update model parameters via backward propagation and gradient descent. 

The system evolves in discrete time slots under resource constraints, where each edge server operates with heterogeneous computational capabilities and energy budgets. 
Let $\mathcal{T}=\lbrace0,1,\ldots, T\rbrace$ and $\tau$ denote the time index set and the duration of each time slot, respectively. 
We consider the one-to-one association between the edge servers and the experts. That is, expert $j$ is deployed at edge server $j$, $j\in \mathcal{J}=\lbrace1, \ldots, J\rbrace$. 
Suppose that tokens arrive at the router according to a Poisson process with an average arrival rate of $\lambda$ tokens per time slot. Let a binary variable $x_{i,j}(t)\in\lbrace0,1\rbrace$ represent the token routing indictor, where $x_{i,j}(t)=1$ if the token $i$ is routed to the expert $j$ at time slot $t$, otherwise $x_{i,j}(t)=0$. 
%In this letter, the router follows the classic top-$K$ routing rule, where the gating network only selects the $K$ most suitable experts. 
The top-$K$ routing strategy yields $\sum_{j\in \mathcal{J}}x_{i,j}(t)=K$, $\forall i\in \mathcal{S}(t),t\in \mathcal{T}$, where $\mathcal{S}(t)$ denotes the token index set at the time slot $t$. 

% Accordingly, the total number of tasks assigned to the $j$th expert at time slot $t$ is given by $d_j^\mathrm{allo}(t)=\sum_{i\in \mathcal{S}(t)}x_{i,j}(t)$.
% \begin{equation}
% \label{deqn_ex1}
% d_j^{allo}(t)=\sum_{i\in \mathcal{S}}x_{i,j}(t).
% \end{equation}
% We assume that the communication latency for dispatching tasks from the BS to the devices is negligible and that the wireless links are reliable (i.e., error-free). This allows us to focus on the computational and energy-related aspects of the system without introducing communication uncertainties.

At each time slot, the tokens from the router to each edge server are first stored in a token queue. Each edge server processes the tokens in its queue following a first-come-first-serve policy.
In this paper, a token is said to be processed if it has been routed to $K$ experts, computed by these experts, and aggregated by the router to produce the final output. 
Specifically, the number of tokens completed by edge server $j$ at time slot $t$, is given by 
\begin{equation}
\label{deqn_ex1}
d_j^{\mathrm{com}}(t)=\mathrm{min}\left\lbrace Q_j(t)+d_j^{\mathrm{rou}}(t),\left\lfloor {\tau}/{\tau_j^{\mathrm{com}}(t)}\right\rfloor \right\rbrace,
\end{equation}
where the floor function $\lfloor \cdot \rfloor$ rounds a number down to the nearest integer, $\tau_j^{\mathrm{com}}=c_i/f_j(t)$ is the computation time required by edge server $j$ to process a single token, $d_j^\mathrm{rou}(t)=\sum_{i\in \mathcal{S}(t)}x_{i,j}(t)$ is the total number of tokens routed to the expert $j$, $c_i$ is the number of CPU cycles for computing every single token, and $f_j(t)$ is the operating CPU frequency of edge server $j$. Note that the duration of each time slot should be significantly longer than the computation time of a single token, i.e., $0\leq \tau _j^\mathrm{com}(t)\leq \tau $. The CPU frequency is constrained by the edge server's maximum computational capability $0\leq f_j(t) \leq f^\mathrm{max}.$ 
Let $Q_j(t)$ be the backlog of the token queue at edge server $j$ over the time slot $t$. The token queuing dynamics at edge server $j$ is given by~\cite{neely2010stochastic}
\begin{equation}
\label{equ:Q_update}
Q_j(t+1)=\mathrm{max}\lbrace Q_j(t)+d_j^\mathrm{rou}(t)-d_j^\mathrm{com}(t), 0\rbrace.
\end{equation}
Note that $0\leq d_j^\mathrm{com}(t)\leq Q_j(t)+d_j^\mathrm{rou}(t).$

At each time slot, each edge server incurs energy consumption due to executing the routed computational tokens. Specifically, the computation energy consumed by edge server $j$ at the time slot $t$ is given by~\cite{mao2016dynamic, deng2020wireless}
\begin{equation}
\label{deqn_ex3}
E_j^\mathrm{com}(t)=\xi_j d_j^\mathrm{com}(t) \tau_j^\mathrm{com}(t) f_j^3(t),
\end{equation}
where $\xi_j$ is the effective switched capacitance determined by the edge server's chip architecture~\cite{burd1996processor}. In practice, the energy budget of each edge server is limited, i.e., $0\leq E_j^\mathrm{com}(t) \leq E_j^\mathrm{max} $. From \eqref{deqn_ex1} and \eqref{deqn_ex3}, the computational capabilities (i.e., the maximum number of tokens computed per time slot) of different edge servers are diverse. To manage and constrain the energy usage of edge servers, we introduce an energy queue that reflects the long-term evolution of energy consumption. The energy queuing dynamics at edge server $j$ is given by
\begin{equation}
\label{equ:E_update}
Z_j(t+1)=\mathrm{max}\lbrace Z_j(t)+E_j^\mathrm{com}(t)-E_j^\mathrm{avg}, 0\rbrace,
\end{equation}
where $Z_j(t)$  represents the energy backlog at the beginning of time slot $t$, and $E_j^\mathrm{avg}$ is the average energy budget available to edge server $j$.
% This queueing model captures the notion that when the actual energy consumption $E_j^\mathrm{com}(t)$ exceeds the sustainable budget $E_j^\mathrm{avg}$ , the energy queue accumulates, indicating an energy deficit. Conversely, when energy consumption falls below the budget, the energy queue decreases and may even reach zero. It is important to note that a zero energy queue is acceptable and signifies that the device is operating with sufficient energy; however, any surplus energy does not carry over for future compensation of energy deficits.

\begin{figure}[!t]
\centering
\includegraphics[width=3.2in]{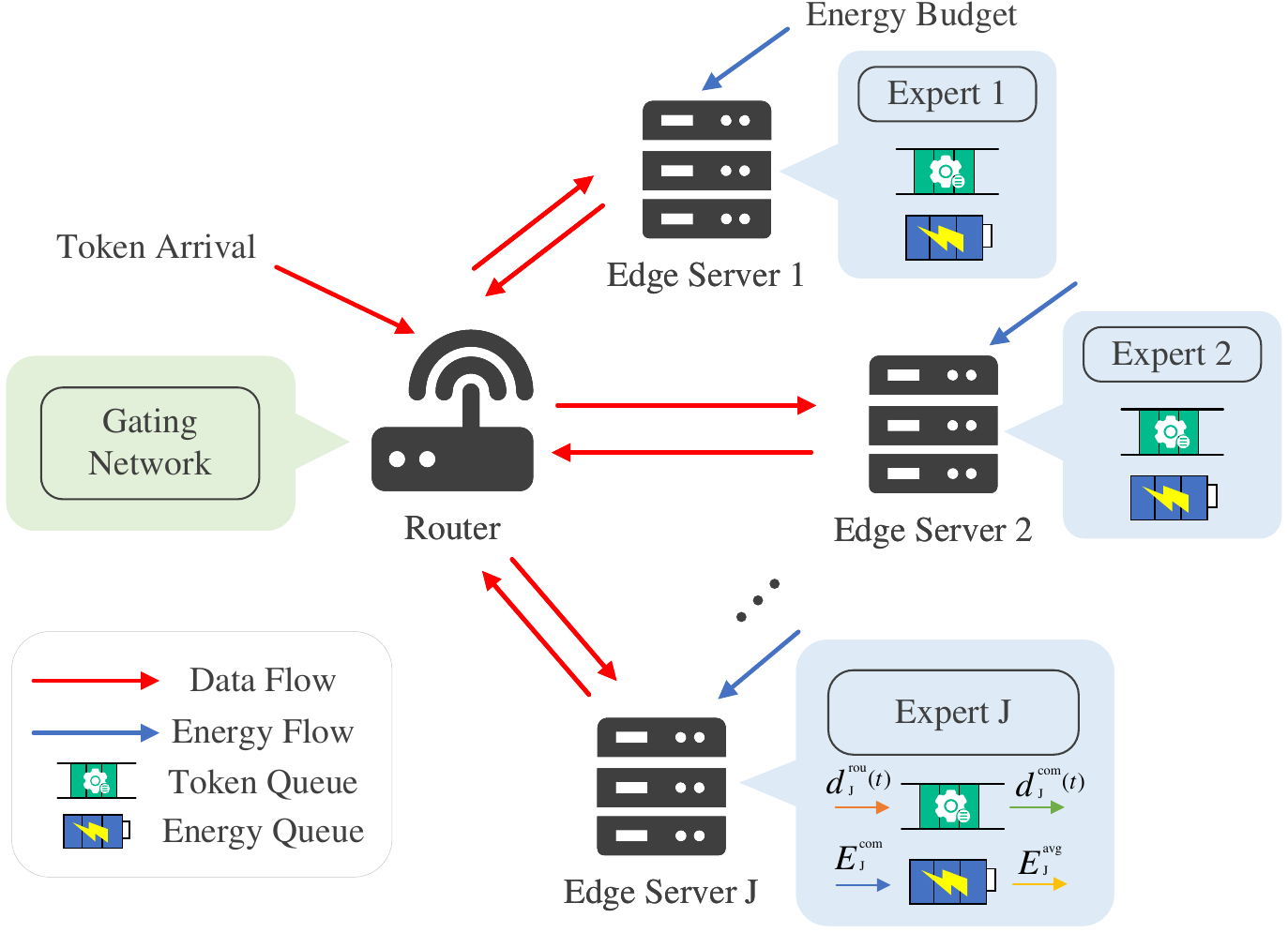}
% \vspace{-10pt}
\caption{The distributed MoE training over edge networks with a router and distributed edge servers.}
\label{fig_1}
\end{figure}

\subsection{Problem Formulation}
%In conventional MoE systems, tokens are typically assigned to experts with the $K$ highest gating scores~\cite{shazeer2017outrageously}. Its token assignment neglects the heterogeneity and real-time workloads of edge servers, which results in severe load imbalance across the edge network. In that case, some experts are frequently selected while others are almost idle.

Under the traditional top-$K$ routing strategy such as~\cite{shazeer2017outrageously}, the gating network scores and activates the experts solely based on token features, but neglects the impact of resource constraints and token backlogs of edge servers. This inevitably results in inefficient resource utilization and thereby degrades system throughput, where some servers are overloaded with excessive tokens while the others remain underutilized. 

Driven by this challenge, the goal of Stable-MoE is to maximize the long-term system throughput and gating consistency by enabling resource-aware token routing and efficient computation frequency utilization of the edge servers. 
In this paper, the system throughput is defined as the number of tokens processed per time slot.  
Let $\boldsymbol{X}(t)=[\boldsymbol{x}(t), \boldsymbol{f}(t)]$, wherein $\boldsymbol{x}(t)=\lbrace x_{i,j}(t)\rbrace _{\mathcal{S}(t) \times J}$ denotes the binary token routing matrix as previously defined, and $\boldsymbol{f}(t)=[f_1(t), f_2(t),\ldots,f_J(t)]$ denotes the computation frequencies of the edge servers. In this context, we formulate the optimization problem as 
% \begin{equation}
% \label{deqn_ex5}
% \begin{aligned}
% \mathbf{P}:&\max_{\boldsymbol{X}(t)}{\gamma(\overline{\boldsymbol{d}})+\mu\overline{G}}\\
% \text{s.t.}\,\,&\text{C}1: \sum_{j\in \mathcal{J}}x_{i,j}(t)=K, x_{i,j}(t)\in \lbrace 0,1\rbrace,\forall i\in \mathcal{S}(t),\\
% &\qquad j\in \mathcal{J},t\in \mathcal{T},\\
% &\text{C}2: 0\leq f_j(t)\leq f_{\mathrm{CPU}}^{\mathrm{max}},\forall j \in \mathcal{J},t\in \mathcal{T},\\
% &\text{C}3: 0\leq \tau _j^{\mathrm{com}}(t)\leq \tau,\forall j \in \mathcal{J},t\in \mathcal{T},\\
% &\text{C}4: 0\leq E_j^{\mathrm{com}}(t)\leq E_j^{\mathrm{max}},\forall j \in \mathcal{J},t\in \mathcal{T},\\
% &\text{C}5: \lim_{t \to \infty}\frac{\mathbb{E}\lbrace |Q_j(t)|\rbrace}{t}=0,\forall j \in \mathcal{J},t\in \mathcal{T},\\
% &\text{C}6: \lim_{t \to \infty}\frac{\mathbb{E}\lbrace |Z_j(t)|\rbrace}{t}=0,\forall j \in \mathcal{J},t\in \mathcal{T},
% \end{aligned}
% \end{equation}
\begin{align}
\mathbf{P}:&\max_{\boldsymbol{X}(t)}\enspace{\gamma(\overline{\boldsymbol{d}})+\mu\overline{G}} \label{deqn_ex5}\\
\text{s.t.}\,\,&\text{C}1: \sum_{j\in \mathcal{J}}x_{i,j}(t)=K, x_{i,j}(t)\in \lbrace 0,1\rbrace,\forall i\in \mathcal{S}(t),\nonumber\\
&\qquad j\in \mathcal{J},t\in \mathcal{T},\nonumber\\
&\text{C}2: 0\leq f_j(t)\leq f^{\mathrm{max}},\forall j \in \mathcal{J},t\in \mathcal{T},\nonumber\\
&\text{C}3: 0\leq \tau _j^{\mathrm{com}}(t)\leq \tau,\forall j \in \mathcal{J},t\in \mathcal{T},\nonumber\\
&\text{C}4: 0\leq E_j^{\mathrm{com}}(t)\leq E_j^{\mathrm{max}},\forall j \in \mathcal{J},t\in \mathcal{T},\nonumber\\
&\text{C}5: \lim_{t \to \infty}\mathbb{E}\lbrace |Q_j(t)|\rbrace /{t}=0,\forall j \in \mathcal{J},t\in \mathcal{T},\nonumber\\
&\text{C}6: \lim_{t \to \infty}\mathbb{E}\lbrace |Z_j(t)|\rbrace /{t}=0,\forall j \in \mathcal{J},t\in \mathcal{T},\nonumber
\end{align}
where $\gamma (\overline{\boldsymbol{d}})=\sum_{j\in \mathcal{J}}\mathrm{log}(1+\overline{d_j^\mathrm{com}})$ is the system utility that positively correlates with the system throughput~\cite{neely2010stochastic}. Note that $\gamma (\overline{\boldsymbol{d}})$ is a concave and increasing function that balances throughput and fairness by discouraging excessive load on individual edge servers.
$\overline{d_j^\mathrm{com}}=\lim_{T \to \infty}\frac{1}{T}\sum_{t=0}^{T-1}\mathbb{E}[d_j^\mathrm{com}(t)]$ is the time average of completed tokens at edge server $j$, and $\overline{\boldsymbol{d}}=[\overline{d_1^\mathrm{com}},\ldots,\overline{d_J^\mathrm{com}}]$. 
We optimize $G(t)=\sum_{i\in  \mathcal{S}(t)}\sum_{j\in \mathcal{J}}g_{i,j}(t) x_{i,j}(t)$ to maximize the consistency between traditional and optimized routing strategies, where $g_{i,j}(t)\in [0,1]$ is the gating scores of traditional MoE. 
Specifically, a larger $G(t)$ ensures that tokens are preferentially routed to experts with higher gating scores, thereby preserving model quality while dynamically adapting to the underlying resource constraints.
Additionally, $\overline{G}=\lim_{T \to \infty}\frac{1}{T}\sum_{t=0}^{T-1}\mathbb{E}[G(t)]$ is the time average of $G(t)$, and $\mu \geq 0$ serves as a weighting factor that controls the trade-off between system throughput and gating consistency. 

Notably, problem $\textbf{P}$ is a stochastic optimization problem that aims to jointly optimize long-term system throughput and routing consistency under stochastic token arrivals and energy constraints. 
% Specifically, the objective function consists of a logarithmic throughput utility term and a soft regularization term for gating consistency. 
Constraints ($\text{C}1$)-($\text{C}4$) have been defined in Section II-A.
Moreover, $\text{C}5$ and $\text{C}6$ ensure long-term stability of the token and energy queues, respectively~\cite{neely2010stochastic}.

\section{Stable-MoE}

It is rather challenging to solve $\textbf{P}$ offline due to coupled decisions, long-term objectives, and system uncertainty. 
To address this, Lyapunov optimization~\cite{neely2010stochastic} offers a principled framework for making sequential decisions in dynamic systems with time-varying states and unknown future information. 
% Unlike conventional optimization approaches that rely on full knowledge of arrival statistics and system distributions, Lyapunov optimization enables online control by observing only the current system state and making decisions slot-by-slot. It guarantees that all long-term queue stability constraints are satisfied while simultaneously driving the system utility toward its optimal region.
% Furthermore, the drift-plus-penalty technique derived from Lyapunov optimization balances two critical aspects: the minimization of queue backlogs (representing system stability) and the maximization of instantaneous utility (representing performance). This makes it particularly well-suited to wireless distributed MoE systems, where edge devices operate under limited energy budgets, fluctuating workloads, and the need to coordinate expert selection decisions.
% Therefore, in the following sections, we leverage Lyapunov optimization to design a low-complexity, online scheduling policy that dynamically allocates tasks and adjusts device computation resources, ensuring both long-term performance and queue stability without requiring future knowledge or distribution assumptions.
In this paper, we use the Lyapunov method to transform the intractable long-term optimization problem into tractable subproblems per time slot. 

For the queuing dynamics $Q_j(t)$ and $Z_j(t)$, we define the Lyapunov function as $L(t)=\frac{1}{2}\lbrace\sum_{j \in \mathcal{J}}(Q_j(t)^2+Z_j(t)^2)\rbrace$. Let $\boldsymbol{\Theta}(t)=\lbrace Q_j(t), Z_j(t),\forall j \in \mathcal{J} \rbrace$ denote the system state at time slot $t$. We define the conditional Lyapunov drift for time slot $t$ as 
% \begin{equation}
% \label{deqn_ex1}
$\Delta(t)=\mathbb{E}[L(t+1)-L(t)|\boldsymbol{\Theta}(t)]$,
% \end{equation}
which depends on the scheduling policy in response to dynamically arriving tokens and the current system state $\boldsymbol{\Theta}(t)$. Minimizing $\Delta(t)$ helps stabilize the virtual queues $\boldsymbol{\Theta}(t)$, thereby promoting satisfaction of the long-term stability constraints $\text{C}5$ and $\text{C}6$~\cite{neely2010stochastic}. However, minimizing $\Delta(t)$ alone may result in low system throughput. 
Instead, we minimize the Lyapunov drift-plus-penalty function $\Delta_V(t)=\Delta(t)-V\mathbb{E}[\gamma(\overline{\boldsymbol{d}})+\mu\overline{G}|\boldsymbol{\Theta}(t)]$, where $V>0$ is a tunable parameter that controls the trade-off between optimization objective and queue stability. 
% The instantaneous utility function at time slot $t$ is defined as $U(t)=\sum_{j\in \mathcal{J}}\mathrm{log}(1+d_j^\mathrm{com}(t))+\mu\sum_{i\in \mathcal{S}(t)}\sum_{j\in \mathcal{J}}g_{i,j}(t) x_{i,j}(t)$, which captures both the throughput-fairness objective and the soft constraint on gating consistency. By minimizing an upper bound of the drift-plus-penalty expression in each time slot, we can derive a control policy that guarantees long-term queue stability while achieving near-optimal scheduling decisions. 

To address the challenge of directly minimizing $\Delta_V(t)$, we follow the Lyapunov optimization and characterize an upper bound of $\Delta_V(t)$ as the following lemma:
\begin{lemma}
For any queue backlogs and actions, $\Delta_V(t)$ is upper bounded by 
\begin{align}
&\Delta_V(t)\leq B-V\mathbb{E}[\sum_{j\in \mathcal{J}}\mathrm{log}(1+d_j^\mathrm{com}(t))+\mu \sum_{i \in \mathcal{S}(t)}\sum_{j\in \mathcal{J}}\nonumber\\
&\quad g_{i,j}(t) x_{i,j}(t)|\boldsymbol{\Theta}(t)]+\sum_{j\in \mathcal{J}}Q_j(t)\mathbb{E}[d_j^\mathrm{rou}(t)-d_j^\mathrm{com}(t)|\boldsymbol{\Theta}(t)]\nonumber\\
&\quad +\sum_{j\in \mathcal{J}}Z_j(t)\mathbb{E}[E_j^\mathrm{com}(t)-E_j^\mathrm{avg}|\boldsymbol{\Theta}(t)], \label{eqn-6}
\end{align}
where
\begin{equation}
\label{deqn_ex7}
B=\frac{1}{2}\sum_{j\in \mathcal{J}}[(\lambda+\lambda^2)+(D^\mathrm{max}_j)^2+(E^\mathrm{max}_j)^2+(E^\mathrm{avg}_j)^2].
\end{equation}
\end{lemma}
\begin{proof}
From (2), we have
\begin{align}
\label{equ:lemma1}
\begin{aligned}
&Q_j(t+1)^2=(\mathrm{max}\lbrace Q_j(t)+d_j^\mathrm{rou}(t)-d_j^\mathrm{com}(t),0\rbrace)^2 \\
&\leq [Q_j(t)+d_j^\mathrm{rou}(t)-d_j^\mathrm{com}(t)]^2 \\
&= Q_j(t)^2+2Q_j(t)[d_j^\mathrm{rou}(t)-d_j^\mathrm{com}(t)]+[d_j^\mathrm{rou}(t)-d_j^\mathrm{com}(t)]^2.
\end{aligned}
\end{align}
By moving $Q_j(t)^2$ to the left-hand side of \eqref{equ:lemma1}, dividing both sides by $2$, and summing up the inequalities $\forall j \in \mathcal{J}$, we have
\begin{align}
\label{equ:lemma2}
&\frac{1}{2}\sum_{j\in \mathcal{J}}[Q_j(t+1)^2-Q_j(t)^2]\leq \sum_{j\in \mathcal{J}}Q_j(t)[d_j^\mathrm{rou}(t)-d_j^\mathrm{com}(t)] \nonumber\\
&+\frac{1}{2}\sum_{j\in \mathcal{J}}[d_j^\mathrm{rou}(t)-d_j^\mathrm{com}(t)]^2.
\end{align}
Similarly, for the energy queue in \eqref{equ:E_update}, we have
\begin{align}
\label{equ:lemma3}
&\frac{1}{2}\sum_{j\in \mathcal{J}}[Z_j(t+1)^2-Z_j(t)^2] \leq \sum_{j\in \mathcal{J}}Z_j(t)[E_j^\mathrm{com}(t)-E_j^\mathrm{avg}] \nonumber\\
&+\frac{1}{2}\sum_{j\in \mathcal{J}}[E_j^\mathrm{com}(t)-E_j^\mathrm{avg}]^2.
\end{align}
Summing up \eqref{equ:lemma2} and \eqref{equ:lemma3}, taking the conditional expectation, we have $\Delta(t) \leq 
\sum_{j\in \mathcal{J}}Q_j(t)\mathbb{E}[d_j^\mathrm{rou}(t)-d_j^\mathrm{com}(t)|\boldsymbol{\Theta}(t)]
+\sum_{j\in \mathcal{J}}Z_j(t)\mathbb{E}[E_j^\mathrm{com}(t)-E_j^\mathrm{avg}|\boldsymbol{\Theta}(t)]
+\sum_{j\in \mathcal{J}}\frac{1}{2}\mathbb{E}[(d_j^\mathrm{rou}(t)-d_j^\mathrm{com}(t))^2+(E_j^\mathrm{com}(t)-E_j^\mathrm{avg})^2|\boldsymbol{\Theta}(t)].$
% \begin{equation}
% \label{equ:lemma4}
% \begin{aligned}
% &\Delta(t) \leq 
% \sum_{j\in \mathcal{J}}Q_j(t)\mathbb{E}[d_j^\mathrm{rou}(t)-d_j^\mathrm{com}(t)|\boldsymbol{\Theta}(t)]\\
% &+\sum_{j\in \mathcal{J}}Z_j(t)\mathbb{E}[E_j^\mathrm{com}(t)-E_j^\mathrm{avg}|\boldsymbol{\Theta}(t)]\\
% &+\sum_{j\in \mathcal{J}}\frac{1}{2}\mathbb{E}[(d_j^\mathrm{rou}(t)-d_j^\mathrm{com}(t))^2+(E_j^\mathrm{com}(t)-E_j^\mathrm{avg})^2|\boldsymbol{\Theta}(t)].\\
% \end{aligned}
% \end{equation}
Then, subtracting $V\mathbb{E}[\sum_{j=1}^J \mathrm{log}(1+d_j^\mathrm{com}(t))+\mu \sum_{i=1}^{|\mathcal{S}(t)|}\sum_{j=1}^Jg_{i,j}(t) x_{i,j}(t)|\boldsymbol{\Theta}(t)]$ from both sides yields 
% \begin{equation}
% \label{equ:lemma4}
% \begin{aligned}
% &\Delta_V(t)\leq \sum_{j\in \mathcal{J}}\frac{1}{2}\mathbb{E}[(d_j^\mathrm{allo}(t)-d_j^\mathrm{com}(t))^2\\
% &+(E_j^\mathrm{com}(t)-E_j^\mathrm{avg}(t))^2|\boldsymbol{\Theta}(t)]-V\mathbb{E}[\sum_{j\in \mathcal{J}}\mathrm{log}(1+d_j^\mathrm{com}(t))\\
% &+\mu \sum_{i \in \mathcal{S}(t)}\sum_{j\in \mathcal{J}}g_{i,j}(t) x_{i,j}(t)|\boldsymbol{\Theta}(t)]\\
% &+\sum_{j\in \mathcal{J}}Q_j(t)\mathbb{E}[d_j^\mathrm{allo}(t)-d_j^\mathrm{com}(t)|\boldsymbol{\Theta}(t)]\\
% &+\sum_{j\in \mathcal{J}}Z_j(t)\mathbb{E}[E_j^\mathrm{com}(t)-E_j^\mathrm{avg}|\boldsymbol{\Theta}(t)].
% \end{aligned}
% \end{equation}
\begin{align}\
&\Delta_V(t)\leq \sum_{j\in \mathcal{J}}\frac{1}{2}\mathbb{E}[(d_j^\mathrm{rou}(t)-d_j^\mathrm{com}(t))^2 \nonumber\\
&+(E_j^\mathrm{com}(t)-E_j^\mathrm{avg})^2|\boldsymbol{\Theta}(t)]-V\mathbb{E}[\sum_{j\in \mathcal{J}}\mathrm{log}(1+d_j^\mathrm{com}(t)) \nonumber\\
&+\mu \sum_{i \in \mathcal{S}(t)}\sum_{j\in \mathcal{J}}g_{i,j}(t) x_{i,j}(t)|\boldsymbol{\Theta}(t)] \nonumber\\
&+\sum_{j\in \mathcal{J}}Q_j(t)\mathbb{E}[d_j^\mathrm{rou}(t)-d_j^\mathrm{com}(t)|\boldsymbol{\Theta}(t)] \nonumber\\
&+\sum_{j\in \mathcal{J}}Z_j(t)\mathbb{E}[E_j^\mathrm{com}(t)-E_j^\mathrm{avg}|\boldsymbol{\Theta}(t)]. \label{equ:lemma5}
\end{align}

Note that $d_j^{\mathrm{rou}}(t), d_j^{\mathrm{com}}(t),$ and $E_j^{\mathrm{com}}(t)$ are non-negative. Thus, we have $\frac{1}{2}\mathbb{E}[(d_j^\mathrm{rou}(t)-d_j^\mathrm{com}(t))^2|\boldsymbol{\Theta}(t)] \leq \frac{1}{2}\mathbb{E}[d_j^\mathrm{rou}(t)^2|\boldsymbol{\Theta}(t)]+\frac{1}{2}\mathbb{E}[d_j^\mathrm{com}(t)^2|\boldsymbol{\Theta}(t)].$ Since $0\leq d_j^\mathrm{rou}(t) \leq |\mathcal{S}(t)|$, it follows that $\mathbb{E}[d_j^\mathrm{rou}(t)^2|\boldsymbol{\Theta}(t)] \leq \mathbb{E}[\mathcal{S}(t)^2]=\lambda+\lambda^2$. Moreover, let $D_j^\mathrm{max}=\lfloor \frac{\tau f_j^\mathrm{max}}{c_i}\rfloor$ be the maximum processing capacity of edge server $j$, we obtain $\mathbb{E}[d_j^{\mathrm{com}}(t)^2|\boldsymbol{\Theta}(t)] \leq (D_j^\mathrm{max})^2$. Similarly, $\mathbb{E}[(E_j^{\mathrm{com}}(t)-E_j^{\mathrm{avg}})^2|\boldsymbol{\Theta}(t)] \leq (E_j^{\mathrm{max}})^2+(E_j^{\mathrm{avg}})^2$. Thus, we arrive at \eqref{eqn-6}.
% Without loss of generality, we assume $|\mathcal{S}(t)|$ is upper bounded by $S^\mathrm{max}=\lambda+3\sqrt{\lambda}$. Let $D_j^\mathrm{max}=\lfloor \frac{\tau f_j^\mathrm{max}}{c_i}\rfloor$ denote the maximum processing capacity of edge server $j$. Finally, we have \eqref{eqn-2}.
\end{proof}
% Since directly minimizing $\Delta_V(t)$ is intractable, Lyapunov optimization framework seeks to minimize an upper bound of $\Delta_V(t)$ in each time slot to derive a tractable control policy. 
Note that the upper bound of $\Delta_V(t)$ involves conditional expectations. To minimize the upper bound of $\Delta_V(t)$, we employ opportunistic expectation minimization of~\cite{neely2010stochastic} to generate the optimal policy. Specifically, in each time slot, we generate an implementable scheduling policy to minimize
\begin{align}\label{deqn_ex12}
&B-V [ \sum_{j\in \mathcal{J}}\mathrm{log}(1+d_j^\mathrm{com}(t))+\mu \sum_{i \in \mathcal{S}(t)}\sum_{j\in \mathcal{J}}g_{i,j}(t) x_{i,j}(t)] \nonumber\\
&+\sum_{j\in \mathcal{J}}Q_j(t)[d_j^\mathrm{rou}(t)-d_j^\mathrm{com}(t)]+\sum_{j\in \mathcal{J}}Z_j(t)[E_j^\mathrm{com}(t) \nonumber\\
&-E_j^\mathrm{avg}],
\end{align}
where \eqref{deqn_ex12} is obtained from the upper bound of $\Delta_V(t)$ by ignoring the expectations. Given the above formulations, we reformulate \textbf{P} into the per-slot tractable optimization problem:
% 原格式
% \begin{align}\label{equ:p1}
% \mathbf{P1}: \max_{\mathbf{X}(t)}&V [\sum_{j\in \mathcal{J}}\mathrm{log}(1+d_j^\mathrm{com}(t))+\mu \sum_{i\in \mathcal{S}(t)}\sum_{j\in \mathcal{J}}g_{i,j}(t) \nonumber\\
% &x_{i,j}(t)]-\lbrace \sum_{j\in \mathcal{J}}Q_j(t)[d_j^\mathrm{rou}(t)-d_j^\mathrm{com}(t)]+\sum_{j\in \mathcal{J}} \nonumber\\
% &Z_j(t)[E_j^\mathrm{com}(t)-E_j^\mathrm{avg}]\rbrace \nonumber\\
% \text{s.t.}\quad &\text{C}1\sim \text{C}4. 
% \end{align}

\vspace{-15pt}
% 漂移惩罚函数形式
\begin{align}\label{equ:p1}
\mathbf{P1}: \max_{\mathbf{X}(t)}&\enspace V \Gamma_\text{u}(\mathbf{X}(t))-\Gamma_\text{p}(\mathbf{X}(t)) \nonumber\\
\text{s.t.}\quad &\text{C}1\sim \text{C}4, 
\end{align}
% \vspace{-10pt}
where 
\begin{align}\label{equ:util}
\Gamma_\text{u}(\mathbf{X}(t))= \sum_{j\in \mathcal{J}}\mathrm{log}(1+d_j^\mathrm{com}(t))\nonumber\\+\mu \sum_{i\in \mathcal{S}(t)}\sum_{j\in \mathcal{J}}g_{i,j}(t)x_{i,j}(t),
\end{align}
\vspace{-15pt}
\begin{align}\label{equ:pena}
\Gamma_\text{p}(\mathbf{X}(t))= \sum_{j\in \mathcal{J}}Q_j(t)[d_j^\mathrm{rou}(t)-d_j^\mathrm{com}(t)]\nonumber\\+\sum_{j\in \mathcal{J}}Z_j(t)[E_j^\mathrm{com}(t)-E_j^\mathrm{avg}].
\end{align}
As discussed above, the per-slot optimization problem $\textbf{P1}$ involves both integer decision variables and a nonlinear objective function. Due to the complex coupling among variables, we employ a tunable solver-based approach based on mixed-integer programming methods (e.g., branch-and-bound) to optimize the token routing and computation frequency allocation at each time slot. The pseudocode of Stable-MoE is summarized in \textbf{Algorithm \ref{algorithm}}.

% \begin{algorithm}
% \caption{Adaptive Resource Allocation for MoE}\label{algorithm}
% Initialize: $ Q_j(t) = 0$, $ Z_j(t) = 0$\\
% \textbf{At the BS side:} \\
% \hspace{0.5em} Receive task batch $\mathcal{S}(t)$\;
% \hspace{0.5em} Observe $ Q_j(t), Z_j(t)$\;
% \hspace{0.5em} Compute $ g_{i,j}(t)$ based on gating network output\;
% \hspace{0.5em} Find $ x_{i,j}(t)$ and $f_j(t)$ by solving $(14)$\;
% \hspace{0.5em} Dispatch expert assignments to devices. \\
% \textbf{At the $j$-th device:} \\
% \hspace{0.5em} Receive assigned tasks $d^{\mathrm{allo}}_{j}(t)$\;
% \hspace{0.5em} Compute $d^{\mathrm{com}}_{j}(t)$ from $f_j(t)$ and $\tau^{\mathrm{com}}_{j}(t)$\;
% \hspace{0.5em} Update $(2)$\;
% \hspace{0.5em} Compute energy $E^{\mathrm{com}}_{j}(t)$\;
% \hspace{0.5em} Update $(4)$\;
% \hspace{0.5em} Execute expert inference and return output. \\
% \end{algorithm}

\begin{algorithm}[h]
\SetAlgoLined
% \KwData{The parameter $V$, the weighting factor $\mu$, the parameter $K$, the duration of time slot $\tau$, the maximum computation frequency $f_{\text{CPU}}^{\text{max}}$, the maximum energy budget $E_{j}^{\text{max}}$, the average energy budget $E_{j}^{\text{avg}}$ }
\KwData{$V$, $\mu$, $K$, $\tau$, $f^{\text{max}}$, $E_{j}^{\text{max}}$, $E_{j}^{\text{avg}}$ }
% \KwResult{The sub-adjacency matrix $\{\bm{A}^{m}, \forall m\in M\}$}
Initialization: $ Q_j(0) = 0$, $ Z_j(0) = 0$, $t = 0$\;
\While{$t < T$}{
Receive token batch $\mathcal{S}(t)$\;
Observe $ Q_j(t), Z_j(t)$\;
Calculate $ g_{i,j}(t)$\;
Find $ x_{i,j}(t)$ and $f_j(t)$ by solving \eqref{equ:p1}\;
\For{each edge server $j \in \mathcal{J}$}{
Receive tokens $d^{\mathrm{rou}}_{j}(t)$\;
Calculate $d^{\mathrm{com}}_{j}(t)$ by $f_j(t)$ and $\tau^{\mathrm{com}}_{j}(t)$, and then update token queuing dynamics as \eqref{equ:Q_update}\;
Calculate $E^{\mathrm{com}}_{j}(t)$, and then update the energy queuing dynamics as \eqref{equ:E_update}\;
Generate and return forward outputs\;
}
Produce the final output and the training loss\;
Update parameters of gating network and experts\;
$t = t+1$\;
}
\caption{Stable-MoE}\label{algorithm}
\end{algorithm}
\vspace{-10pt}

\section{Experimental Results}
We evaluate the performance of Stable-MoE in resource-constrained edge networks. The experiments are conducted on two benchmark datasets: SVHN~\cite{netzer2011reading} and CIFAR-100~\cite{krizhevsky2009learning}, both consisting of 32$\times$32 colored digit images. 
The gating network is feedforward layers, and the expert networks are convolutional layers. 
We generate the tokens at time slot $t$ by sampling a batch of images from the dataset according to a Poisson distribution with parameter $\lambda$. This sampling strategy ensures balanced label distribution while maintaining the stochastic characteristics of real-time token arrivals. Consider that $J=10$ edge servers with $K=3$, time slot duration $\tau =1\mathrm{s}$, average token arrival rate $\lambda =390 \mathrm{tokens/slot}$, effective switched capacitance $\xi_j=2\times10^{-27}$, maximum CPU frequency $f^{\mathrm{max}}=3 \mathrm{GHz}$, computational complexity $c_i=10^7 \mathrm{cycles/token}$. 
We set ${E_{j}^{\mathrm{max}}}\in [3\mathrm{J}, 15\mathrm{J}]$ and ${E_{j}^{\mathrm{avg}}}\in [1.5\mathrm{J}, 9.5\mathrm{J}]$, $\forall j\in \mathcal{J}$. These non-uniform energy budgets intrinsically constrain $f_j(t)$ and correspondingly limit $d_j^{\text{com}}(t)$ of each expert, thereby leading to the heterogeneous computational capabilities of edge servers.

For comparison purpose, we consider the following baselines: (a) random routing (see the line labeled with ``Strategy A''): tokens are randomly routed among all experts with uniform probability; (b) traditional top-$K$ routing (see the line labeled with ``Strategy B''): tokens are routed to the experts with the $K$ highest gating scores; (c) queue-aware routing (see the line labeled with ``Strategy C''): tokens are routed to the experts with the smallest token queue backlog at each time slot; (d) energy-aware routing (see the line labeled with ``Strategy D''): tokens are routed to the experts with the smallest energy queue backlog. 

\begin{figure}[!t]
\centering
\includegraphics[width=3.5in]{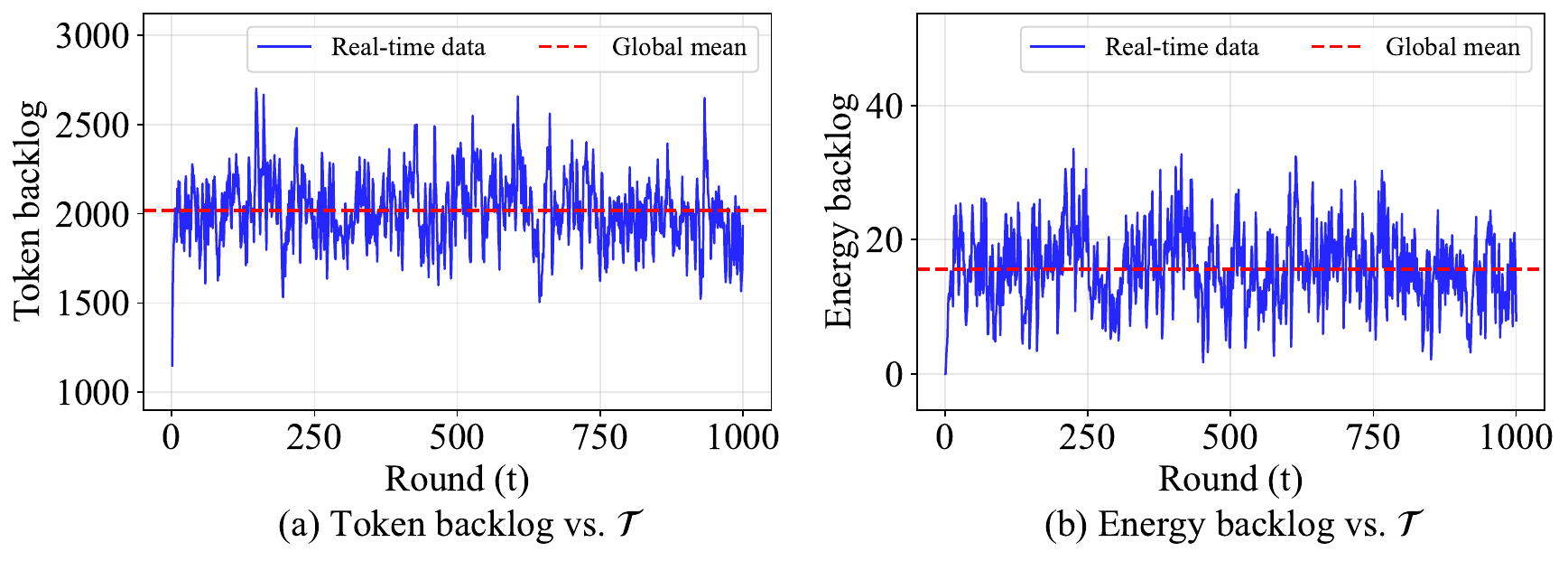}
% \vspace{-20pt}
\caption{Queue backlogs of Stable-MoE versus $\mathcal{T}$. The blue line represents the instantaneous backlog observed in each round, whereas the red dashed lines indicate the global mean values obtained by averaging over all rounds.}
\label{fig:queue}
\end{figure}
% \vspace{-5pt}
\begin{figure}[!t]
\centering
\includegraphics[width=3.5in]{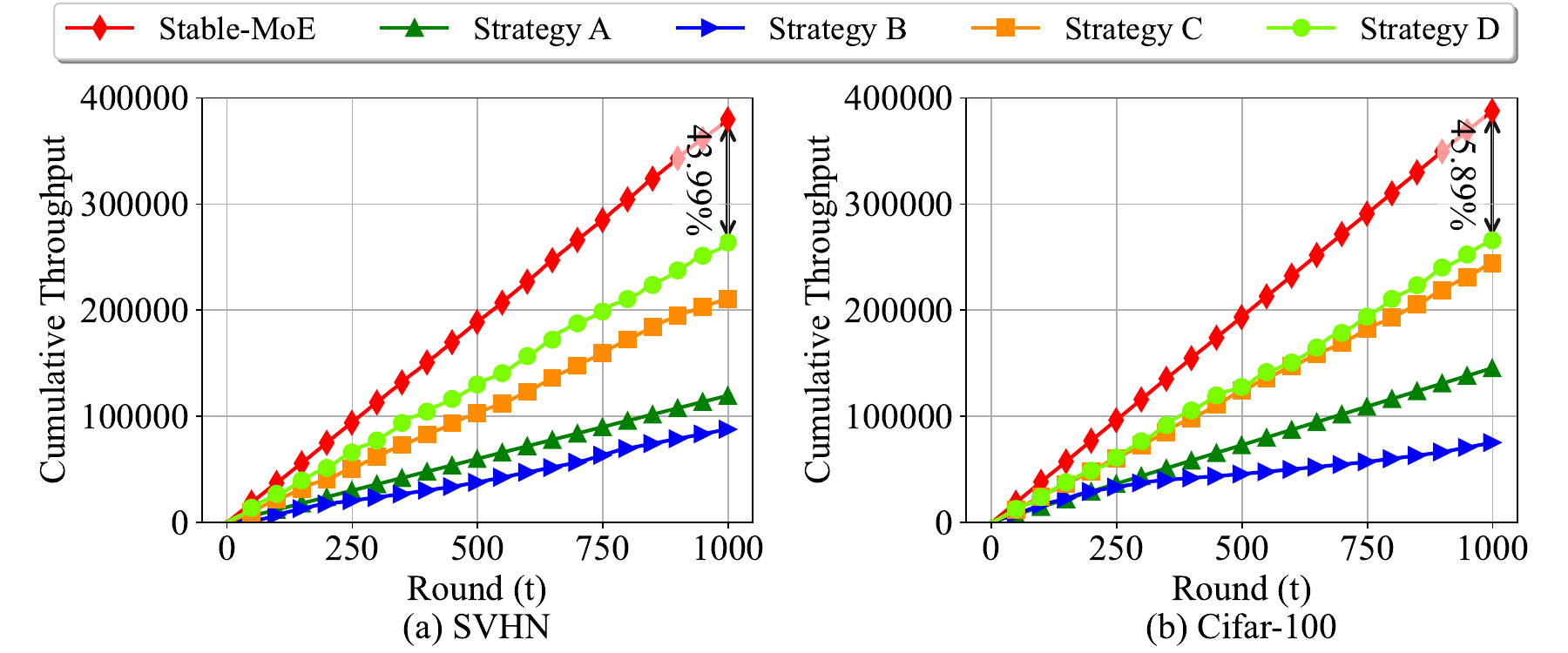}
% \vspace{-2pt}
\caption{Throughput comparison between Stable-MoE and Strategies A-D.}
\label{fig:throughput}
\end{figure}
% \vspace{-12pt}
Fig. \ref{fig:queue} shows the temporal evolution of token and energy queue backlogs versus $\mathcal{T}$ under Stable-MoE. It is observed that both token and energy queues exhibit initial growth followed by stabilization around steady-state values. These results further demonstrate the queue stability of Stable-MoE under stochastic token arrivals and heterogeneous energy constraints.

Fig. \ref{fig:throughput} compares the cumulative system throughput over time of Stable-MoE against four baselines under SVHN and CIFAR-100. It is observed that, Stable-MoE achieves significantly higher cumulative throughput than Strategies A-D, and the performance gain continuously widens over time. 
% For example, the performance gain of Stable-MoE over Strategy D is around $40\%$. 
This is due to the fact that Strategy D focuses solely on minimizing energy queues without considering the computational workload. As a result, edge servers with smaller energy queues may still be overloaded.  

\begin{figure}[!t]
\centering
\includegraphics[width=3.5in]{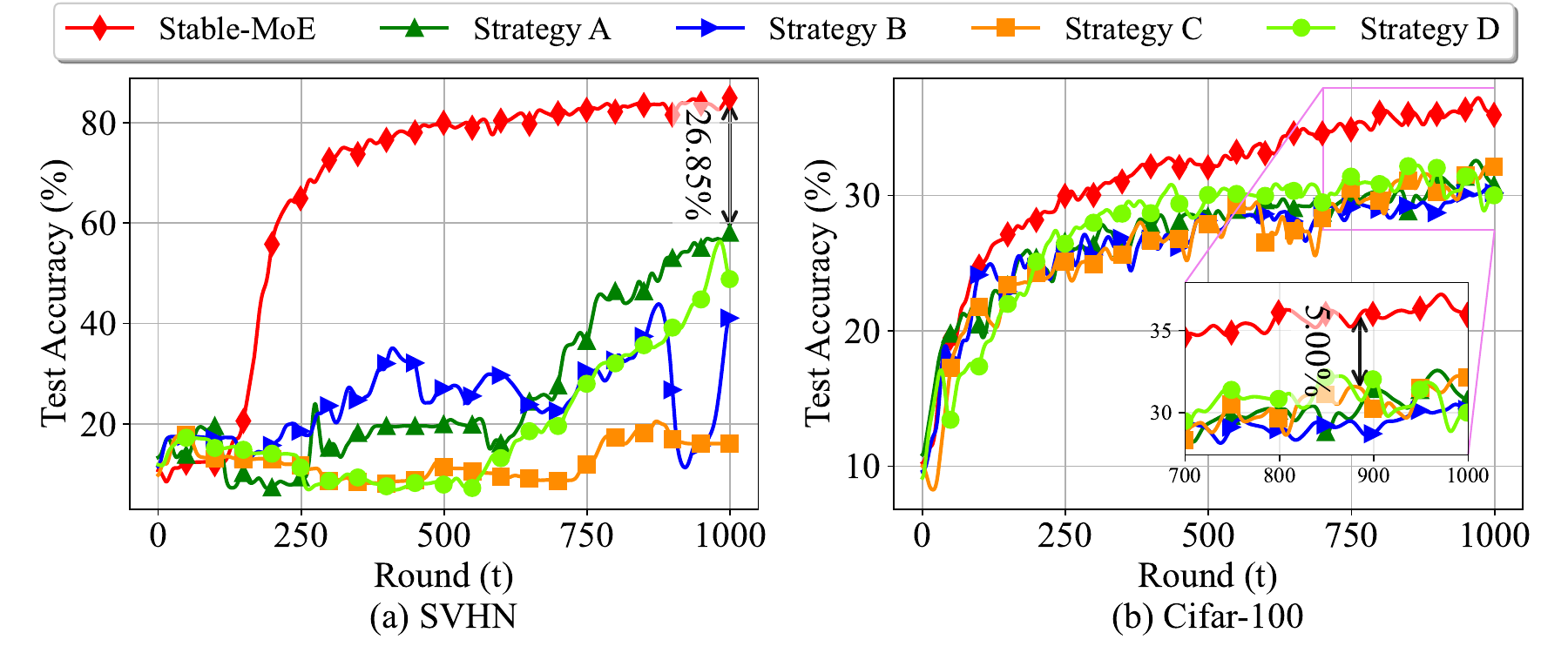}
% \vspace{-20pt}
\caption{Accuracy comparison between Stable-MoE and Strategies A-D.}
\label{fig:accuracy}
\end{figure}

Fig. \ref{fig:accuracy} presents the test accuracy for Stable-MoE and baselines. 
In Fig. \ref{fig:accuracy}(a), the accuracy of Stable-MoE rapidly converges to $80\%$, which significantly outperforms all the baselines. 
In contrast, the baselines either converge slowly or remain unstable with persistently low accuracy. 
% Moreover, Strategy B exhibits the poor and unstable accuracy, since it ignores the diverse computation capacities of edge servers. 
In Fig. \ref{fig:accuracy}(b), it is observed that Stable-MoE gains the best accuracy. 
% It is observed that LLBMoE consistently achieves the highest accuracy and converges the fastest among all methods. The rapid convergence and superior accuracy highlight the effectiveness of LLBMoE in distributed training. In contrast, Strategies A, C, and D achieve moderate accuracy levels but with slower convergence and higher variance. 
This result further demonstrates the high efficiency of Stable-MoE in resource-heterogeneous edge networks. 

\section{Conclusion}
We have proposed a novel Lyapunov-based token routing framework for distributed MoE systems over resource-heterogeneous edge networks to dynamically optimize token routing and computation frequency allocation.
Given stochastic token arrivals, Stable-MoE ensures the long-term stability of both token and energy queues under diverse computational capabilities of distributed edge servers. 
Experimental results have demonstrated that Stable-MoE achieves superior throughput compared with baseline strategies while maintaining queue stability. 
Overall, this paper serves as a first foray into the dynamic token routing issue of resource-heterogeneous MoE from the perspective of Lyapunov optimization. 
Going forward, future directions include the designs of dynamic routing strategies under sophisticated model architecture and over resource-limited wireless networks.

\bibliographystyle{IEEEtran}
\bibliography{references}  % references.bib文件

% \bibitem{ref1}
% {\it{Mathematics Into Type}}. American Mathematical Society. [Online]. Available: https://www.ams.org/arc/styleguide/mit-2.pdf

% \bibitem{ref2}
% T. W. Chaundy, P. R. Barrett and C. Batey, {\it{The Printing of Mathematics}}. London, U.K., Oxford Univ. Press, 1954.

% \bibitem{ref3}
% F. Mittelbach and M. Goossens, {\it{The \LaTeX Companion}}, 2nd ed. Boston, MA, USA: Pearson, 2004.

% \bibitem{ref4}
% G. Gr\"atzer, {\it{More Math Into LaTeX}}, New York, NY, USA: Springer, 2007.

% \bibitem{ref5}M. Letourneau and J. W. Sharp, {\it{AMS-StyleGuide-online.pdf,}} American Mathematical Society, Providence, RI, USA, [Online]. Available: http://www.ams.org/arc/styleguide/index.html

% \bibitem{ref6}
% H. Sira-Ramirez, ``On the sliding mode control of nonlinear systems,'' \textit{Syst. Control Lett.}, vol. 19, pp. 303--312, 1992.

% \bibitem{ref7}
% A. Levant, ``Exact differentiation of signals with unbounded higher derivatives,''  in \textit{Proc. 45th IEEE Conf. Decis.
% Control}, San Diego, CA, USA, 2006, pp. 5585--5590. DOI: 10.1109/CDC.2006.377165.

% \bibitem{ref8}
% M. Fliess, C. Join, and H. Sira-Ramirez, ``Non-linear estimation is easy,'' \textit{Int. J. Model., Ident. Control}, vol. 4, no. 1, pp. 12--27, 2008.

% \bibitem{ref9}
% R. Ortega, A. Astolfi, G. Bastin, and H. Rodriguez, ``Stabilization of food-chain systems using a port-controlled Hamiltonian description,'' in \textit{Proc. Amer. Control Conf.}, Chicago, IL, USA,
% 2000, pp. 2245--2249.

% \end{thebibliography}

\vfill

\end{document}